\documentclass{article}
\usepackage[all]{xy}
\usepackage{amsthm}
\usepackage{amsfonts}
\usepackage{amsmath}
\usepackage{amssymb}
\usepackage{mathrsfs}
\usepackage{graphicx}
\usepackage{subcaption}
\usepackage{url}
\usepackage{nomencl} % Remove in the final version
\usepackage{todonotes} % Remove in the final version

\theoremstyle{definition}
\newtheorem{definition}{Definition}
\newtheorem{proposition}{Proposition}
\newtheorem{corollary}{Corollary}

\newtheorem{theorem}{Theorem}
\newcommand{\R}{\mathbb{R}}
\newcommand{\C}{\mathbb{C}}

\newcommand{\Q}{\mathbb{Q}}
\newcommand{\T}{\mathbb{T}}
\newcommand{\Z}{\mathbb{Z}}
\newcommand{\X}{\mathbb{X}}
\newcommand{\N}{\mathbb{N}}
\newcommand{\ii}{\mathrm{i}}
\newcommand{\dd}{\,\mathrm{d}}

\DeclareMathOperator{\rank}{rank}
\DeclareMathOperator{\tr}{tr}

\makeatletter
\let\@fnsymbol\@alph
\makeatother

\title{Constraints on pure point diffraction on aperiodic point patterns of finite local complexity}
\author{Pavel Kalugin\footnote{Laboratoire de Physique des Solides, CNRS, Universit\'e Paris-Sud, Universit\'e Paris-Saclay, F-91405 Orsay, France. E-mail: \texttt{kalugin@lps.u-psud.fr}}
	\and Andr\'e Katz\footnote{Directeur de recherche honoraire, CNRS, France}}

\date{}
\begin{document}
\maketitle
\begin{abstract}
    It is shown that the partial amplitudes of the pure point part of the diffraction spectrum of an aperiodic Delone point pattern of finite local complexity are linked by a set of linear constraints. These relations can be explicitly derived from the geometry of the prototile space of the underlying tiling.
\end{abstract}
\section{Introduction}
Delone sets of finite local complexity (FLC) are traditional objects of study in the diffraction theory of aperiodic solids. In this paper, instead of considering individual Delone sets which just {\em happen to} have the FLC property, we rather deal with families of such sets, having common {\em allowed} local configurations, and study the constraints on the pure point diffraction stemming from these local rules. The main motivation of this approach comes from the problem of the structural analysis in quasicrystals and is explained in details in \cite{kalugin2019robust}. However, in contrast with \cite{kalugin2019robust}, this paper is not limited to the case of the quasiperiodic long-range order. Neither the restrictions on local environments need to  have the strength of the {\em matching rules}, that is fix by themselves the long-range order of the structure. In particular, the results are also applicable to the pure point part of diffraction of random tiling models.
\par
The results of this paper are formulated in terms of the partial diffraction amplitudes of the following distribution associated with an FLC Delone multiset $(\Lambda_1,\dots,\Lambda_m)$ in the $d\mbox{-dimensional}$ Euclidean space $E$:
\begin{equation}
\label{eq:density}
\varrho=\sum_{p=1}^m \left(\sum_{y \in \Lambda_p} w_p \delta_y\right),
\end{equation}
where the index $p$ enumerates atomic sites distinguished by their local environment and the weights $w_p \in \C$ represent their diffractive power. The partial amplitudes will be properly defined in Proposition \ref{def_a_k}, but it is reasonable to think of them informally as of the following quantities:
\begin{equation}
\label{eq:naive}
a_{k,p} = \lim_{r \to \infty} \frac{1}{\mathrm{vol}(D_r)} 
\sum_{y \in \Lambda_p \cap D_r} w_p \exp(-2\pi \ii k \cdot y),
\end{equation}
where $k \in E^*$ is the wave vector corresponding to a Bragg peak and $D_r \subset E$ is the ball of radius $r$ centered at the origin. The hypothesis that the pure-point part of the diffraction measure $\eta$ of the distribution (\ref{eq:density}) is related to the amplitudes (\ref{eq:naive}) by the formula
\begin{equation}
\label{eq:contib_naive}
\eta(\{k\}) = \left| \sum_{p=1}^m w_p a_{k,p}\right|^2
\end{equation}
is commonly known as Bombieri-Taylor conjecture \cite{bombieri1986distributions}.
\par
We shall consider the sets $\Lambda_p$ in (\ref{eq:density}) as decorations of a tiling in $E$, and treat the corresponding prototile space as the embodiment of the local rules. In this setting, the informal expression (\ref{eq:naive}) offers a glimpse of the nature of the constraints on the partial diffraction amplitudes $a_{k,p}$, considered as functions on the prototile space. Indeed, as moving the $p\mbox{-th}$ decorating point within a prototile shifts the entire set $\Lambda_p$ and thus modifies (2) by a phase factor, all these functions belong to a finite-dimensional linear space. Furthermore, since the decoration at a tile boundary can be assigned to either of the neighboring tiles, the partial amplitudes are subject to additional linear constraints akin to the Kirchhoff's current conservation law. 
\par
The paper is organized as follows. Section \ref{sec:model} sets up the model of Delone multisets of finite local complexity in terms of flat-branched semi-simplicial complexes and isometric windings. Section \ref{sec:partial} is devoted to a formal definition of the partial diffraction amplitudes in terms of dynamical systems. The main results of the paper (Theorem \ref{one_peak} and Corollary \ref{cor:smooth}) are exposed in Sections \ref{sec:single} and \ref{sec:subspace}. Finally, Section \ref{sec:examples} illustrates these results by several examples of aperiodic point patterns.
\section{The model}
\label{sec:model}
A convenient way to impose local rules on the multiset $(\Lambda_1,\dots,\Lambda_m)$ in (\ref{eq:density}) consists in treating it as a decoration of a tiling of $E$. The range of the local rules is not limited by the tiles sizes, since the environment on a longer range can always be specified by discrete labels attached to the tiles. Without loss of generality, we can limit the consideration to the case of simplicial tilings. Note, that contrary to a common usage, we treat such a tiling as a {\em partition} of $E$ by a countable set of interiors of affine simplices, in particular we accept simplicial tiles of any dimension from $0$ to $d$ (by a common convention, a vertex is its own interior). A class of tiles identically labeled and having the same shape (up to a translation) is called a prototile (this term may also refer to an arbitrary representative of the class). One can introduce an order on the vertices of each prototile by choosing a direction in $E$ not orthogonal to any prototile edge, and ordering vertices by the values of their projections on this direction. Then the $i\mbox{-th}$ face of a prototile is defined as its face not containing the $i\mbox{-th}$ vertex. Let $B_n$ stand for the set of prototiles of dimension $n$. The matching constraints on the tiles are realized by defining the maps $\delta_{n,i}: B_n \to B_{n-1}$ assigning to each $s \in B_n$ the prototile $\delta_{n,i}s \in B_{n-1}$ glued to its $i\mbox{-th}$ face. Since the ordering of vertices remains consistent across dimensions, the maps $\delta_{n,i}$ satisfy the simplicial identity  (\ref{face_identity}). We shall denote the resulting semi-simplicial set (see Appendix A) by $B$. The corresponding geometric realization $|B|$ thus represents the prototile space of the tiling.
\par
In addition to the combinatorial data encoded by $B$, the local rules must also specify the geometry of prototiles, which can be recovered from the directions and orientations of edges. The latter are constrained by the condition that the edges of each $2\mbox{-dimensional}$ face of a prototile form a triangle. This constraint can be conveniently formulated in terms of the chain complex $C_\bullet(B, \Z)$ (see Appendix A), leading to the following definition \cite{kalugin2019robust}:
\begin{definition} \label{FBS}
	A $d\mbox{-dimensional}$ flat-branched semi-simplicial complex (FBS-complex) is a triple $(B, E, \rho)$, where $B$ is a finite semi-simplicial set of dimension $d$, $E$ is a $d\mbox{-dimensional}$ real Euclidean vector space and $\rho$ is a homomorphism $C_1(B, \Z) \to E$ satisfying the following conditions:
	\begin{itemize}
		\item The homomorphism $\rho$ vanishes on boundaries. 
		\item For any $s \in B_d$, the vectors $\rho(e_{d,1}(s)), \dots, \rho(e_{d,d}(s))$ are linearly independent.  
	\end{itemize} 
\end{definition}
For any $s \in B_n$ let $\sigma_s \subset E$ be given by the formula
\begin{equation}
\label{eq:sigma}
\sigma_s := \left\{\sum_{i=1}^{n} c_i \rho(e_{n,i}(s)) \,\,\middle|\, 
c_i \in \R^+ \text{ and }
\sum_{i=1}^{n} c_i<1 \right\} 
\end{equation}
The second condition of Definition \ref{FBS} guarantees that $\sigma_s$ is a non-degenerate $n\mbox{-dimensional}$ affine simplex. Identification of barycentric coordinates within $|s| \subset |B|$ and $\sigma_s$ defines a homeomorphism
\begin{equation}
\label{alpha_s}
\alpha_s: |s| \to \sigma_s.
\end{equation}
The tilings obeying the matching rules encoded by an FBS-complex are in one-to-one correspondence with the {\em isometric windings} of the latter \cite{kalugin2019robust}:
\begin{definition}
	A continuous map $f:E \to |B|$ is called isometric winding of an FBS-complex $(B, E, \rho)$ if
	\begin{itemize}
		\item For each $s \in B$ the restriction $f\big|_{f^{-1}(|s|)}$ is a covering map of $|s|$.
		\item The composition $\alpha_s \circ f\big|_{\sigma}$ restricted to any connected component $\sigma$ of $f^{-1}(|s|)$ is a translation of $\sigma$ by some vector of $E$.
	\end{itemize}
\end{definition}
The tiling $\mathcal{T}_f$ corresponding to an isometric winding $f:E \to |B|$ is given by the partition of $E$ by connected components of $f^{-1}(|s|)$ for all $s \in B$. Fixing a point $x \in |B|$ defines a decoration of this tiling by the set $f^{-1}(x)$. We shall use this construction to define the sets $\Lambda_p$ in (\ref{eq:density}) by fixing $m$ points $x_p \in |B|$ and setting $\Lambda_p := f^{-1}(x_p)$. This leads to the following formula for the distribution of the diffracting quantity:
\begin{equation}
\label{varrho}
\varrho_f:=\sum_{p=1}^m \left(\sum_{y \in f^{-1}(x_p)} w_p \delta_y\right),
\end{equation}
The formula (\ref{varrho}) describes the model of the distribution of matter which will be used throughout the rest of the paper.
\section{Partial diffraction amplitudes}
\label{sec:partial}
The main flaw of the informal expression (\ref{eq:naive}) is the presence of the limit operation, which makes proving any result about $a_{k,p}$ a daunting task. The crucial step towards a closed expression for the diffraction amplitudes consists in using the theory of dynamical systems. The key element in this scheme is the {\em hull} of the diffracting distribution -- a compact topological space representing arbitrarily large finite patches of an infinite system ``all at once''. The idea to use the hull in studying diffraction has been proposed by Dworkin in \cite{dworkin1993spectral}, and since then has lead to a significant progress in understanding the relation between the diffraction and dynamical spectra (\cite{deng2008dworkin,lenz2017stationary,baake2016spectral}). In this section, we follow mostly the ideas of \cite{deng2008dworkin} and \cite{baake2004dynamical}, but with the twist of using a matrix-valued diffraction measure.
\par
In the theory of aperiodic order, hulls are built by adding limiting points to the orbit of an aperiodic structure under the action of translations. The actual construction may come in different guises, either as a closure of the orbit in an appropriate topological vector space (for the hulls of almost periodic functions or measures), or as a completion of the orbit in appropriate metric (in the case of point sets or tilings \cite{sadun2008topology}). Since we shall be mostly interested in the dependence of the pure point component of the diffraction measure on the parameters $x_p$ and $w_p$ in (\ref{varrho}) (cf. the formula (\ref{eq:peak}) below), the construction of the hull should depend only on the properties of the isometric winding $f$ in (\ref{varrho}). The translation of the diffracting distribution by a vector $t \in E$ corresponds to the following transformation of $f$: 
\begin{equation}
\label{T_t_Hull}
T_t f: y \mapsto f(y-t).
\end{equation}
For a given isometric winding $f_0: E \to |B|$ we shall define its hull $\X(f_0)$ as the closure of its orbit in the compact-open topology of $C(E, |B|)$:
\begin{equation}
\label{eq:def_hull}
\X(f_0):=\overline{\left\{ T_t f_0 \mid t \in E \right\}}
\end{equation}
Let us show that this construction is equivalent to that of the tiling space (a continuous hull) \cite{sadun2008topology} of $\mathcal{T}_{f_0}$. Consider the patch of a shifted tiling $\mathcal{T}_{f_0}+t$ contained within a compact window $K \subset E$. Since $\mathcal{T}_{f_0}$ has finite local complexity, there exists a compact $K' \subset E$ such that all these patches are generated by shifts $t \in K'$. Therefore
\begin{equation}
\label{eq:compact}
\left\{
(T_t f_0)\big|_{K} \mid t \in E 
\right\} 
=
\left\{
(T_t f_0)\big|_{K} \mid t \in K' 
\right\} 
\end{equation}
Since the right-hand side of (\ref{eq:compact}) is an image of the compact $K'$, the left-hand side is closed in the compact-open topology of $C(K, |B|)$. Thus, for any $f \in \X(f_0)$ there exists $t \in E$ such that $f\big|_K = (T_t f_0)\big|_K$. Therefore all points of $\X(f_0)$ are isometric windings and thus correspond to tilings, while the convergence in $\X(f_0)$ is equivalent to that in the tiling space of $\mathcal{T}_{f_0}$.
\par
The closed formula for the partial diffraction amplitudes requires one more ingredient -- a translation invariant probability measure $\mu$ on $\X(f_0)$. Such a measure always exists since $\X(f_0)$ is compact, but it may be not unique. Although the results of this section remain valid for any translation invariant measure, it should be emphasized that the well-definedness of the limit in the formula (\ref{eq:naive}) is guaranteed only in the case when there exists only one such measure, that is when the measure-preserving dynamical system $(\X(f_0), E, \mu)$ is uniquely ergodic \cite{lenz2009continuity}. On the other hand, in real physical applications the condition of unique ergodicity is not very stringent since it corresponds to an intuitive notion of macroscopic uniformity of the specimen. Bearing this in mind we shall treat the measure $\mu$ as a background parameter and omit to reference it unless necessary.
\par
Let $\mathcal{S}(E)$ stand for the Schwartz space on $E$. For a given point $x \in |B|$, let $\Gamma_x$ be the linear map $\mathcal{S}(E) \to L^2(\X(f_0), \mu)$ defined by the formula
\begin{equation}
\label{Gamma}
\Gamma_x(\varphi):= \left(
f \mapsto \sum_{y \in {f}^{-1}(x) }\varphi(-y)
\right) \quad\text{ where } \varphi \in \mathcal{S}(E).
\end{equation}
Since $\varphi$ is rapidly decreasing and ${f}^{-1}(x)$ is uniformly discrete, the sum in the above expression converges absolutely, therefore $\Gamma_x$ is continuous. Consider the following sesquilinear functional on $\mathcal{S}(E)$ (we use the convention that Dirac bracket is antilinear in the first argument):
$$
(\varphi_1, \varphi_2) \mapsto \sum_{p,q=1}^m
\overline{w_p} w_q \left\langle 
\Gamma_{x_p}(\varphi_1), \Gamma_{x_q}(\varphi_2)
\right\rangle.
$$ 
This functional is translationally invariant, positive definite and continuous in each argument. Therefore (see \cite[Chapter II.3, Theorem 6]{gel2014generalized} and the discussion afterwards), there exists a positive tempered measure $\eta$ on $E^*$ such that
\begin{equation}
\label{eta}
\sum_{p,q = 1}^m
\overline{w_p} w_q \left\langle 
\Gamma_{x_p}(\varphi_1), \Gamma_{x_q}(\varphi_2)
\right\rangle = \int_{E^*} \overline{\widehat{\varphi_1}}(k) \widehat{\varphi_2}(k) \dd\eta(k)
\end{equation}
We shall take (\ref{eta}) as the definition of the diffraction measure $\eta$ of the distribution (\ref{varrho}) (see Appendix B for the proof that this definition is equivalent to more traditional ones).
\par
The diffraction measure $\eta$ in (\ref{eta}) is sesquilinear in the weights $w_p$:
\begin{equation}
\label{etazeta}
\eta=\sum_{p,q=1}^m \overline{w_p} w_q \zeta_{pq},
\end{equation}
where $\zeta_{pq}$ are complex-valued tempered measures on $E^*$. It is convenient to consider them as elements of a matrix-valued measure on $E^*$ 
$$
\zeta := (\zeta_{pq}) \text{ where } 1\le p,q\le m.
$$
The values of $\zeta$ on bounded Borel sets of $E^*$ are Hermitian positive semi-definite $m \times m$ matrices. 
\par
We shall now follow the ideas of \cite{deng2008dworkin} and construct the isometric embedding of Hilbert spaces
\begin{equation}
\label{Theta}
\Theta: L^2(E^*, \zeta; \C^m) \to L^2(\X(f_0), \mu),
\end{equation}
where  $L^2(E^*, \zeta; \C^m)$ stands for the space of (classes of) functions $E^* \to \C^m$ square-integrable with respect to the matrix-valued diffraction measure $\zeta$ (see \cite{rosenberg1964} for the details). Let as start by defining the action of $\Theta$ on the Schwartz space $\mathcal{S}(E^*, \C^m) \subset L^2(E^*, \zeta; \C^m)$ by the formula
\begin{equation}
\label{ThetaS}
\Theta(\widehat{\varphi} \otimes e_p) := \Gamma_{x_p}(\varphi),
\end{equation}
where $\varphi \in \mathcal{S}(E)$ and $(e_p)$ stands for the canonical basis in $\C^m$. As follows from (\ref{eta}) and (\ref{etazeta}), $\Theta$ intertwines the inner product of $L^2(E^*, \zeta; \C^m)$ restricted to $\mathcal{S}(E^*, \C^m)$ with that of $L^2(\X(f_0), \mu)$. The isometric embedding (\ref{Theta}) is then defined as the continuous extension the map (\ref{ThetaS}) to $L^2(E^*, \zeta; \C^m)$, which is unique by virtue of the following proposition:
\begin{proposition}
	\label{prop:density}
	$\mathcal{S}(E^*, \C^m)$ is dense in $L^2(E^*, \zeta; \C^m)$
\end{proposition}
\begin{proof}
	Let us denote the scalar measure given by the trace of $\zeta$ by $\tr(\zeta)$ and consider $L^2(E^*, \tr(\zeta))\otimes \C^m$ as a space of classes of $\C^m\mbox{-valued}$ functions on $E^*$. Since $\zeta$ is non-negative, the norm of $L^2(E^*, \tr(\zeta))\otimes \C^m$ is stronger than that of $L^2(E^*, \zeta; \C^m)$. This defines a continuous linear map
	\begin{equation}
	\label{trace_inclusion}
	L^2(E^*, \tr(\zeta))\otimes \C^m \to L^2(E^*, \zeta; \C^m).
	\end{equation}	
	As follows from \cite[Theorem 3.11]{rosenberg1964}, simple functions are dense in $L^2(E^*, \zeta; \C^m)$. By standard arguments so are also the simple functions with compact support. Since $\tr(\zeta)$ is locally finite, the latter belong to the image of (\ref{trace_inclusion}), which is therefore dense in $L^2(E^*, \zeta; \C^m)$. Let $C_c(E^*, \C^m)$ stand for the space of continuous $\C^m\mbox{-valued}$ functions with compact support. These functions are approximated by those of $\mathcal{S}(E^*, \C^m)$ uniformly, and thus also in the norm of $L^2(E^*, \tr(\zeta))\otimes \C^m$. It remains to show that $C_c(E^*, \C^m)$ is dense in $L^2(E^*, \tr(\zeta))\otimes \C^m$, which follows immediately from \cite[Theorem 3.14]{rudin1985real}.
\end{proof}
\par
Let us show now that the isometric embedding (\ref{Theta}) is equivariant with respect to translations of $E$. For any $t\in E$ let $S_t$ stand for the unitary operator on $L^2(E^*, \zeta; \C^m)$ defined by the formula
\begin{equation}
\label{eq:S_t}
\left( S_t g \right)(k) := \exp(2 \pi \ii k \cdot t) g(k), 
\end{equation}
where $g \in L^2(E^*, \zeta; \C^m)$ and $k \in E^*$. The translation $t$ also acts on $L^2(\X(f_0), \mu)$ by the unitary operator $T_t$:
\begin{equation}
\nonumber
\left( T_t h \right)(f) := h(T_{-t}f), 
\end{equation}
where $h \in L^2(\X(f_0), \mu)$, $f \in \X(f_0)$  and $T_{-t}f$ is given by the formula (\ref{T_t_Hull}). In particular, for   $\Gamma_x(\varphi) \in L^2(\X(f_0), \mu)$ in (\ref{Gamma}) we have
\begin{equation}
\label{eq:translation}
\left(T_t\Gamma_x(\varphi)\right) (f) =
\Gamma_x(\varphi)(T_{-t}f)=
\Gamma_x(T_{-t}\varphi)(f),
\end{equation}
where $(T_t\varphi)(y):=\varphi(y-t)$. Therefore, as follows from (\ref{ThetaS})
\begin{equation}
\nonumber
\Theta(S_t(\widehat{\varphi} \otimes e_p)) = \Gamma_{x_p}(T_{-t}\varphi)= T_t \Gamma_{x_p}(\varphi) =
T_t (\Theta(\widehat{\varphi} \otimes e_p))
\end{equation}
Then, by virtue of Proposition \ref{prop:density}, the identity
\begin{equation}
\label{eq:intertwine}
\Theta S_t = T_t \Theta
\end{equation}
holds on the entire Hilbert space $L^2(E^*, \zeta; \C^m)$. In other words, the isometric embedding (\ref{Theta}) intertwines the action of $E$ on $L^2(E^*, \zeta; \C^m)$ by $S_t$ with that on $L^2(\X(f_0), \mu)$ by $T_t$.
\par
The notable consequence \cite{deng2008dworkin,baake2016spectral,baake2004dynamical} of the identity (\ref{eq:intertwine}) is that the pure point part $\mathcal{B}$ of the diffraction measure $\zeta$ is a subset of the pure-point part $\mathcal{E} \subset E^*$ of the dynamical spectrum of $(\X(f_0), E, \mu)$, which is defined as the set of values $k \in E^*$ for which there exists a non-zero eigenfunction $\psi_k \in L^2(\X(f_0), \mu)$:
\begin{equation}
\nonumber 
T_t \psi_k = \exp(2 \pi \ii k \cdot t) \psi_k \qquad \text{for any } t \in E
\end{equation}
\par
While the dependence of the diffraction measure $\eta$ on the weights $w_p$ in (\ref{varrho}) is captured by the matrix-valued measure $\zeta$, the latter still depends on the positions of the atomic decorations $\{x_1, \dots, x_m\} \subset |B|$. The following result shows that the pure-point part of the diffraction measure can be expressed in terms of individual contributions of each atomic site $x_p$, thus justifying the formula (\ref{eq:contib_naive}):
\begin{proposition}
	\label{def_a_k}
	Let $f_0:E \to |B|$ be an isometric winding and let $\mu$ be a translationally invariant probability measure on its hull $\X(f_0)$. Then for any eigenvalue $k \in \mathcal{E}$ and a corresponding normalized eigenfunction $\psi_k \in L^2(\X(f_0), \mu)$ there exists a (not necessarily continuous) function $a_k: |B| \to \C$ such that for any set of atomic decorations $\{x_1, \dots, x_m\} \subset |B|$ in (\ref{varrho}) holds the identity
    \begin{equation}
    \label{c_k}
    \left\langle \psi_k, \Theta(1_{\{k\}}\otimes e_p) \right \rangle = a_k(x_p).
    \end{equation}
	Moreover, if $\mu$ is ergodic, one also has
	\begin{equation}
	\label{eq:peak}
	\eta(\{k\})=\left|\sum_{p=1}^m w_p a_{k}(x_p)\right|^2.
	\end{equation}
	We shall refer to the functions $a_k$ as the partial diffraction amplitudes of the dynamical system $(\X(f_0), E, \mu)$.
\end{proposition}
\begin{proof}
    Let us consider the distribution $A_{k,x} \in \mathcal{S}'(E)$ defined by the formula
    $$
    A_{k,x}(\varphi) := \langle \psi_k, \Gamma_x(\varphi) \rangle.
    $$
    Using (\ref{eq:translation}) we get
    \begin{equation}
    \nonumber
    \langle\psi_k, \Gamma_x(T_{t}\varphi)\rangle=
    \langle T_{t}\psi_k, \Gamma_x(\varphi)\rangle,
    \end{equation}
    hence $A_{k,x}$ satisfies the following equation: 
	\begin{equation}
	\label{A_kx}
	A_{k,x}(T_t\varphi)=
	\exp(-2\pi \ii k \cdot t)A_{k,x}(\varphi).
	\end{equation}
	The solutions of this equation in $\mathcal{S}'(E)$ have the form
	$$
	\varphi \mapsto a \int_E \varphi(y) \exp(-2\pi \ii k \cdot y) \dd y,
	$$
   	where $a$ is an arbitrary complex constant. Therefore there exists a function $a_k: |B| \to \C$ such that
	\begin{equation}
	\label{a_k}
	\langle \psi_k,  \Gamma_x(\varphi)\rangle = a_k(x) \widehat{\varphi}(k)
	\end{equation}
	for any $\varphi \in \mathcal{S}(E)$ and any $x \in |B|$. Then as follows from (\ref{ThetaS})
	\begin{equation}
	\label{akxp}
	\left\langle \psi_k, \Theta(\widehat{\varphi}\otimes e_p) \right \rangle = a_k(x_p)\widehat{\varphi}(k)
	\end{equation}
	Since the multiplication by $1_{\{k\}}$ in $L^2(E^*, \zeta; \C^m)$ is the projector on the ei\-gen\-space of $S_t$ (\ref{eq:S_t}), the embedding $\Theta$ intertwines it with the projector on the corresponding eigenspace of $T_t$ and we have
    \begin{equation}
    \nonumber 
    \left\langle \psi_k, \Theta(1_{\{k\}}\widehat{\varphi}\otimes e_p )\right \rangle
    =
    \left\langle \psi_k, \Theta(\widehat{\varphi}\otimes e_p )\right \rangle.
    \end{equation}
	Combining this identity with (\ref{akxp}) yields (\ref{c_k}). 
	\par
	If $\mu$ is ergodic, all eigenvalues of the dynamical system $(\X(f_0), E, \mu)$ are simple and 
	\begin{equation}
	\nonumber
	\Theta(1_{\{k\}}\otimes e_p) = \langle \psi_k, \Theta(1_{\{k\}}\otimes e_p)\rangle\psi_k = a_k(x_p) \psi_k.
	\end{equation}
	Since $\Theta$ is an isometry, this yields
	$$
	\zeta_{pq}(\{k\})=
	\left\langle 1_{\{k\}} \otimes e_p, 1_{\{k\}}\otimes  e_q \right\rangle= \overline{a_k(x_p)} a_k(x_q),
	$$
	and recalling (\ref{etazeta}) we finally get (\ref{eq:peak}).
\end{proof} 
\section{Constraints on partial amplitudes of a single Bragg peak}
\label{sec:single}
We are now going to study the behavior of the partial amplitudes $a_k$ as functions on $|B|$. The results can be conveniently formulated in terms of semi-simplicial vector spaces (see Appendix A)\footnote{In the case when the semi-simplicial set $B$ describes a regular cellular complex the results of Sections \ref{sec:single} and \ref{sec:subspace} can also be formulated in terms of homology of cellular cosheaves \cite{curry2014sheaves} on $|B|$. While this approach might be more natural for the case of no-simplicial tilings, it is not immediately applicable to some relevant examples when the cellular complex associated with $B$ is not regular.}. Let us start by constructing a family of such spaces parameterized by a vector $k \in E^*$.  
\par
Let $\mathcal{F}_{(k),n}$ be the space of functions $|B|\to \C$ spanned by  $\{Y_s: s\in B_n\}$, where
\begin{equation}
\label{eq:Y_s}
Y_s(x):=\begin{cases}
\exp(-2 \pi \ii k \cdot \alpha_s(x)) &\mbox{if } x \in |s| \\
0 & \mbox{if } x \notin |s|
\end{cases}
\end{equation}
Since all functions $\{Y_s: s\in B_\bullet\}$ are linearly independent, they form a basis of the direct sum
$$
\mathcal{F}_{(k), \bullet}:=\bigoplus_{n=0}^d \mathcal{F}_{(k),n},
$$
which is naturally a subspace of the space of all complex-valued functions on $|B|$. 
\begin{proposition}
	\label{prop:in_F}
	Let $a_k$ be the partial diffraction amplitudes defined in Proposition \ref{def_a_k}. Then for any $k \in \mathcal{E}$
	\begin{equation}
	\label{in_sheaf}
	a_k \in \mathcal{F}_{(k), \bullet}.
	\end{equation}
\end{proposition}
\begin{proof}
	Consider a simplex $s\in B_\bullet$ and two points $x_1, x_2 \in |s|$ and let
	$$
	t=\alpha_s(x_2)-\alpha_s(x_1).
	$$
	Then $f^{-1}(x_2)=f^{-1}(x_1)+t$ for any isometric winding $f: E \to |B|$ and 
	$$
	\Gamma_{x_2}(\varphi)=  
	\Gamma_{x_1}(T_t\varphi)
	$$
	for any $\varphi \in \mathcal{S}(E)$. Therefore, as follows from (\ref{A_kx}),
	$$
	a_k(x_2)=
	\exp(-2 \pi \ii k\cdot t)
	a_k(x_1),
	$$
	which proves (\ref{in_sheaf}).
\end{proof}
Let us now provide the graded vector space $\mathcal{F}_{(k),\bullet}$ with the face operators $\delta_{n,i}: \mathcal{F}_{(k), n} \to \mathcal{F}_{(k), n-1}$, defined via the face maps $\delta_{n,i}: B_n \to B_{n-1}$ as:
\begin{equation}
\label{eq:face_F}
\delta_{n,i}(Y_s)=\exp(-2 \pi \ii k \cdot t_{s,i}) Y_{\delta_{n,i}s}
\end{equation}
where $n=\dim(s)$ and $t_{s,i} \in E$ is given by
\begin{equation}
\label{eq:tsi}
t_{s,i}:=\begin{cases}
0 &\mbox{if } i\neq 0 \\
\rho(e_{n,1}(s)) & \text{if } i=0
\end{cases}
\end{equation}\\
\begin{figure}[ht]
	\centering
	\includegraphics[width=0.8\linewidth]{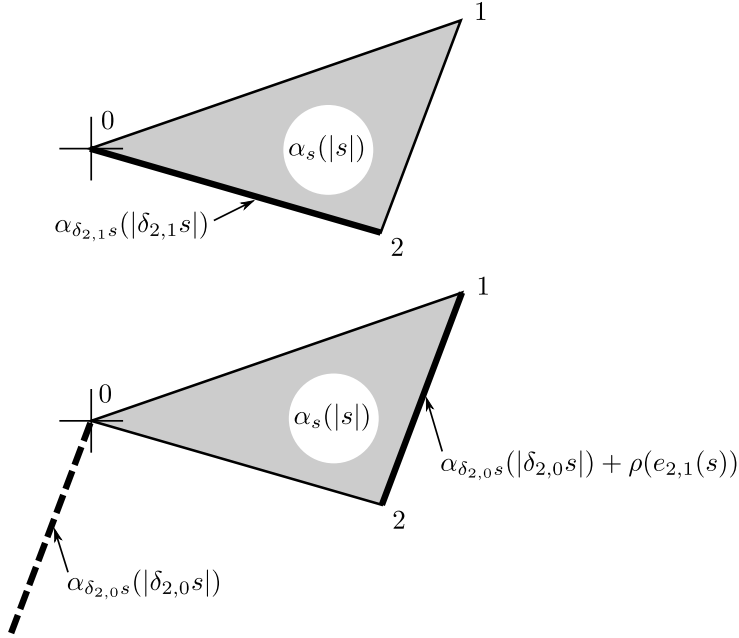}
	\caption{The affine simplex $\alpha_s(|s|) \subset E$ and its $i\mbox{-th}$ face (bold solid line). The cross-hair represents the origin of $E$, numbers are the indices of the vertices of $|s|$. Top: in the case $i \neq 0$, the zeroth vertex of $|\delta_{n,i}s|$ coincides with the zeroth vertex of $|s|$, thus $\alpha_{\delta_{n,i}s}(|\delta_{n,i}s|)$ is the $i\mbox{-th}$ face of $\alpha_s(|s|)$. Bottom: in the case $i=0$, the zeroth vertex of $|\delta_{n,i}s|$ coincides with the first vertex of $|s|$, and $\alpha_{\delta_{n,0}s}(|\delta_{n,0}s|)$ (dashed line) has to be shifted by $\rho(e_{n,1}(s))$ to obtain the zeroth face of $\alpha_s(|s|)$.}
	\label{fig:boundary}
\end{figure}
Note that with the expression (\ref{eq:tsi}) for $t_{s,i}$, the affine simplex $\alpha_{\delta_{n,i}s}(|\delta_{n,i}s|)+t_{s,i}$ is the  $i\mbox{-th}$ face of the affine simplex $\alpha_s(|s|)$ for any $0\le i\le d$ (see Figure \ref{fig:boundary}). Consequently, the linear operator $\delta_{n,i}$ in (\ref{eq:face_F}) acts on the basis function $Y_s$ by its continuation to the $i\mbox{-th}$ face of $|s|$. Therefore, the operators $\delta_{n,i}$ satisfy the simplicial identity (\ref{face_identity}) and provide  $\mathcal{F}_{(k), \bullet}$ with the structure of a semi-simplicial vector space which we shall denote by $\mathcal{F}_{(k)}$.
\par
To formulate the main result of this section, we shall need an orientation of $d\mbox{-simplices}$ $\mathfrak{S}: B_d \to \{-1,1\}$, which can be introduced by fixing a non-zero constant $d\mbox{-form}$ $\Omega \in \bigwedge^d E^*$:
\begin{equation}
\label{eq:sign}
\mathfrak{S}(s):=\mathrm{sgn}\left(\Omega \cdot \bigwedge_{i=1}^d \rho(e_{d, i}(s))\right)
\end{equation}
Let $\mathfrak{s}: |B| \to \{-1,0,1\}$ be defined as
\begin{equation}
\nonumber
\mathfrak{s}(x) := \begin{cases}
\mathfrak{S}(s) & \mbox{if } x \in |s|\mbox{ where } s \in B_d\\
0 & \mbox{otherwise}\end{cases}
\end{equation}
\begin{theorem}
	\label{one_peak}
	Let $a_k$ be the partial diffraction amplitudes defined in Proposition \ref{def_a_k}. Then for any $k \in \mathcal{E}$, the function $\mathfrak{s}a_k$ is a $d\mbox{-cycle}$ of the chain complex $(\mathcal{F}_{(k), \bullet}, \partial_\bullet)$:
    \begin{equation}
    \label{eq:cycle}
    \partial_d (\mathfrak{s}a_k)=0.
    \end{equation}
\end{theorem}
\begin{proof}
    Note first that since $a_k \in \mathcal{F}_{(k), \bullet}$ by Proposition \ref{prop:in_F} and $\mathfrak{s}$ is constant on all $d\mbox{-simplices}$ of $|B|$ and zero elsewhere, we have $\mathfrak{s}a_k \in \mathcal{F}_{(k), d}$. To prove the cycle condition (\ref{eq:cycle}) it suffices to show that $\partial_d (\mathfrak{s}a_k)$ vanishes on any $(d-1)\mbox{-simplex}$ $|s|\subset |B|$. As follows from (\ref{eq:face_F}), the values of $\partial_d (\mathfrak{s}a_k)$ on $|s|$ depend only on the values of $\mathfrak{s}a_k$ on the neighboring $d\mbox{-simplices}$ of $|B|$. Let us denote the set of those neighbors (considered together with the index of the face corresponding to $s$) by $N_s$:  
    \begin{equation}
    \nonumber
    N_s := \left\{
    (s',i) \in B_d \times \{0,\dots,d\} \mathrel{}\middle|\mathrel{}
    \delta_{d,i}s' = s
    \right\}
    \end{equation}
    As follows from (\ref{eq:Y_s}) and (\ref{eq:face_F}), for any $(s', i) \in N_s$ and for any points $x \in |s|$ and $x' \in |s'|$ one has
    \begin{equation}
    \label{eq:two_points}
    \left(\delta_{n,i}(Y_s')\right)(x) = \exp\left(-2 \pi \ii k \cdot(\alpha_s(x)-\alpha_{s'}(x')+ t_{s',i})\right) Y_{s'}(x')
    \end{equation}
    The identity (\ref{eq:two_points}) allows to express the value of $a_k$ at a point $x \in |s|$ via the values of $a_k$ at arbitrarily chosen points $x_{s', i} \in |s'|$ (one for each $(s', i) \in N_s$):
    \begin{multline}
    \label{eq:arb_points}
    \left(\partial_d (\mathfrak{s}a_k)\right)(x) =\\
    \sum_{(s', i) \in N_s} (-1)^i \mathfrak{S}(s') a_k(x_{s', i}) 
    \exp\left(-2 \pi \ii k \cdot(\alpha_s(x)-\alpha_{s'}(x_{s', i})+ t_{s',i})\right)
    \end{multline}
    To finish the proof, we shall use an appropriate choice of the points $x_{s', i}$ to show that that the right-hand side of (\ref{eq:arb_points}) is zero.
    \par
    For any $(s',i) \in N_s$, the affine simplex $\alpha_{s'}(|s'|)-t_{s', i}$ has $\alpha_s(|s|)$ as its $i\mbox{-th}$ face. Let us denote the opposite vertex of $\alpha_{s'}(|s'|)-t_{s', i}$ by $v_{s', i}$. The hyperplane of $E$ containing $\alpha_s(|s|)$ divides the set of points $\{v_{s', i} | (s', i) \in N_s \}$ in two parts, following the sign of the expression 
    \begin{equation}
    \label{eq:side}
    \Omega \cdot\left(
    v_{s', i} \wedge \left(
    \bigwedge_{j=1}^{d-1} \rho(e_{d-1, j}(s))
    \right)
    \right)
    \end{equation}
    The set $N_s$ is thus naturally split as  $N_s = N_s^+ \sqcup N_s^-$, according to the sign of (\ref{eq:side}). As follows from (\ref{eq:tsi}), the position of $v_{s', i}$ is given by the formula
    \begin{equation}
    \nonumber
    v_{s', i}= \begin{cases}
    -\rho(e_{d,1}(s')) &\mbox{if } i= 0 \\
    \rho(e_{d,i}(s')) & \text{if } i \neq 0
    \end{cases}
    \end{equation}
    To compute the sign of (\ref{eq:side}) one can use the identity
    \begin{equation}
    \nonumber
    v_{s', i} \wedge \left(
    \bigwedge_{j=1}^{d-1} \rho(e_{d-1, j}(s))
    \right)
    =
    (-1)^{i+1}\bigwedge_{j=1}^{d} \rho(e_{d, j}(s'))
    \end{equation}
    together with (\ref{eq:sign}), which gives rise to 
    \begin{equation}
    \nonumber
    N_s^\pm = \left\{
    (s',i) \in N_s \mathrel{}\middle|\mathrel{}
    (-1)^{i+1}\mathfrak{S}(s')=\pm 1
    \right\}
    \end{equation}
    Let $\Sigma_s^+$ and $\Sigma_s^-$ stand for the intersections of the affine simplices  $\alpha_{s'}(|s'|)-t_{s', i}$ for $(s', i)$ belonging to $N_s^+$ and $N_s^-$ respectively:
    \begin{equation}
    \label{eq:sigmas}
    \Sigma_s^\pm = \bigcap_{(s', i) \in N_s^\pm} \left(\alpha_{s'}(|s'|)-t_{s', i}\right)
    \end{equation}
    \begin{figure}[ht]
    	\centering
    	\includegraphics[width=0.85\linewidth]{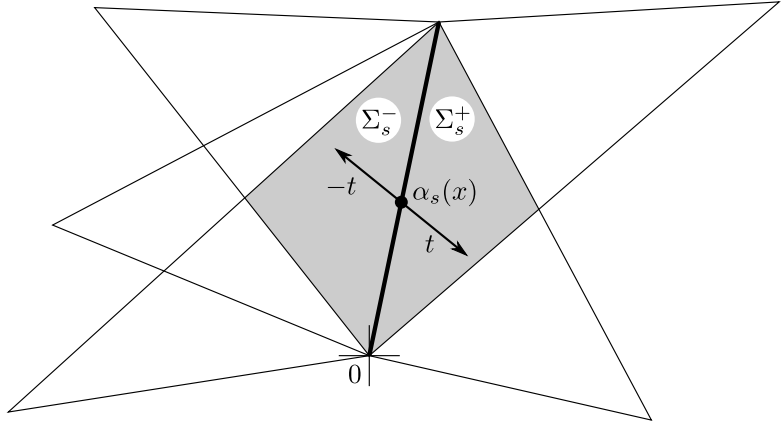}
    	\caption{The superposition of tiles (triangles outlined by thin solid lines) which may have as a face the
    		affine simplex $\alpha_s(|s|) \subset E$ of dimension $d-1$ (bold solid line). The cross-hair represents the origin of $E$. Shaded areas correspond to the open polyhedra $\Sigma_s^+$ and $\Sigma_s^-$ defined in (\ref{eq:sigmas}). Given a point $x \in |s|$, one can always choose a vector $t \in E$ in (\ref{eq:pmt}) in such a way that $\alpha_s(x) \pm t \in \Sigma_s^\pm$.}
    	\label{fig:sigmas}
    \end{figure}
    (see Figure \ref{fig:sigmas}). Both $\Sigma_s^+$ and $\Sigma_s^-$ are non-empty open polyhedra having $\alpha_s(|s|)$ as a face. Therefore, one can choose a vector $t \in E$ such that $\alpha_s(x) \pm t \in \Sigma_s^\pm$ and fix the points $x_{s', i}$ in (\ref{eq:arb_points}) in such a way such that
    \begin{equation}
    \label{eq:pmt}
    \alpha_{s'}(x_{s', i})- t_{s',i} = \alpha_s(x) \pm t \qquad\text{ for } (s', i) \in N_s^\pm
    \end{equation}
    The formula (\ref{eq:arb_points}) then yields
    \begin{equation}
    \label{eq:cancellation}
    \left(\partial_d (\mathfrak{s}a_k)\right)(x) =
    \exp\left(-2 \pi \ii k \cdot t\right) \sum_{x' \in X^-} a_k(x') -
    \exp\left(2 \pi \ii k \cdot t\right) \sum_{x' \in X^+} a_k(x')
    \end{equation}
    where 
    $$
    X^\pm := \left\{x_{s', i} \mathrel{}\middle|\mathrel{} (s', i) \in  N_s^\pm\right\}
    $$
    Note now that as follows from (\ref{eq:pmt}), $f(y)=x$ if and only if $f(y\pm t) \in X^{\pm}$  for any isometric winding $f: E \to |B|$. In other words,
    $$
    f^{-1}(X^\pm) = f^{-1}(x) \pm t
    $$
    and for any $\varphi \in \mathcal{S}(E)$
    $$
    \sum_{x' \in X^\pm} \Gamma_{x'}(\varphi) = \Gamma_x(T_{\pm t}\varphi).
    $$
    Then, as follows from (\ref{A_kx})
    \begin{equation}
    \label{eq:codim}
    \exp\left(\pm 2 \pi \ii k \cdot t\right) \sum_{x' \in X^\pm} a_k(x') = a_k(x)
    \end{equation}
    and the terms at the right-hand side of (\ref{eq:cancellation}) cancel each other, which proves (\ref{eq:cycle}).
\end{proof}
As follows from the equations (\ref{in_sheaf}) and (\ref{eq:cycle}), the functions $a_k$ restricted to the $d\mbox{-dimensional}$ simplices of $|B|$ belong to a linear space of dimension not exceeding $\rank(H_d(\mathcal{F}_{(k), \bullet}, \partial_\bullet))$. On the other hand, as can be seen from (\ref{eq:codim}), the values of $a_k$ on the simplices of dimension $d-1$ depends linearly on those on the neighboring $d\mbox{-dimensional}$ simplices. The reasoning leading to (\ref{eq:codim}) can be easily generalized to simplices of lower dimensions. Therefore, if the number of atomic positions $m$ in the model (\ref{varrho}) is big enough, the contributions of these positions in the pure point diffraction are subject to linear constraints:
\begin{corollary}
	There exist at least $m-\rank(H_d(\mathcal{F}_{(k), \bullet}, \partial_\bullet))$ linear constraints on the partial amplitudes $a_k(x_p)$.
\end{corollary}
\section{Bragg peaks densely filling a subspace}
\label{sec:subspace}
Let $V \subset E^*$ be a non-zero linear subspace of $E^*$. We are going now to establish a connection between constraints on partial amplitudes of different Bragg peaks in the case where $\mathcal{B}$ is dense in an open subset of $V$, and this is the second main result of this paper. Let us start by showing that the semi-simplicial vector spaces $\mathcal{F}_{(k)}$ for any $k \in V$ can be obtained by an appropriate extension of scalars from only one semi-simplicial module.
\par
Let us denote by $L_V \subset V^*$ the image of the group homomorphism $\nu: H_1(B, \Z) \to V^*$ given by the formula 
\begin{equation}
\label{eq:L_V}
\nu(c) := \left(k \mapsto k \cdot \rho_*(c)\right)\qquad\text{ for } c \in H_1(B, \Z)
\end{equation}
For $k \in V$, let $\phi_{k,V}$ stand for the ring homomorphism $\Z[L_V] \to \C$ extending the following character of $L_V$:
\begin{equation}
\label{eq:phi_kV}
l \mapsto \exp(-2 \pi \ii k\cdot l) \qquad \text{ for }l\in L_V
\end{equation}
and let $\phi_{k,V!}$ stand for the functor of extension of scalars by $\phi_{k,V}$.
\begin{proposition}
	\label{prop:iso}
	There exists a semi-simplicial $\Z[L_V]\mbox{-module}$ $\mathcal{G}_V$ such that for any $k \in V$ the semi-simplicial vector space $\phi_{k,V!}\mathcal{G}_V$ is isomorphic to $\mathcal{F}_{(k)}$.
\end{proposition}
\begin{proof}
	We shall define $\mathcal{G}_V$ by providing an {\em ordinary} free graded $\Z[L_V]\mbox{-module}$  $\Z[L_V]^{B_\bullet}$ over the formal basis $(\epsilon_s)_{s \in B_\bullet}$ with the face homomorphisms. Let us start by fixing an element $s_0 \in B_0$ and associating with each $s \in B_\bullet$ an arbitrarily chosen $1\mbox{-chain}$ $c_s \in C_1(B, \Z)$ satisfying the condition
	\begin{equation}
	\label{eq:c_s}
	\partial c_s = \delta_{1,1}( e_{n,1} (s)) - s_0
	\end{equation}
	(this is always possible since $B$ is connected). The face homomorphisms of $\mathcal{G}_V$ are defined via the face maps $\delta_{n,i}: B_n \to B_{n-1}$ as follows:
	\begin{equation}
	\label{eq:face_op}
	\delta_{n,i}(\epsilon_s) = l_{s,i} \epsilon_{\delta_{n,i}(s)}
	\end{equation}
	where $\dim(s)=n$ and $l_{s,i} \in L_V$ is given by
	\begin{equation}
	\nonumber
	l_{s,i} :=  \begin{cases}
	\nu(c_s-c_{\delta_{n,i}s}) & \mbox{if } i \neq 0\\
	\nu(c_s-c_{\delta_{n,i}s} + e_{n,1}(s)) & \mbox{if } i=0\end{cases}
	\end{equation}
	(this expression is well-defined since in both cases the arguments of $\nu$ are cycles). To check that $\mathcal{G}_V$ is indeed a semi-simplicial $\Z[L_V]\mbox{-module}$, it remains to verify that the homomorphisms (\ref{eq:face_op}) satisfy the simplicial identity (\ref{face_identity}). This result stems from the following identity in $L_V$: 
	\begin{equation}
	\label{eq:face_lv}
	l_{\delta_{n,j}s,i} + l_{s,j} = l_{\delta_{n,i}s,j-1} + l_{s,i} \qquad \text{if } i<j
	\end{equation}
	(note that we use the multiplicative notation for the action of $l_{s,i}$ as an element of the group ring $\Z[L_V]$ in (\ref{eq:face_op}) and the additive notation for the group operation in $L_V$ in (\ref{eq:face_lv})). 
	\par
	Let $(\epsilon'_s)_{s \in B_\bullet}$ be the basis in $\phi_{k,V!}\mathcal{G}_V$ corresponding to $(\epsilon_s)_{s \in B_\bullet}$. The functorial image of the face homomorphisms is then given by the following formula:
	$$
	\phi_{k,V!}\delta_{n,i}(\epsilon'_s) = \phi_{k,V}(l_{s,i}) \epsilon'_{\delta_{n,i}(s)} =\exp(-2 \pi \ii k \cdot l_{s,i}) \epsilon'_{\delta_{n,i}(s)}.
	$$
	Consider now the bijective linear map $\omega_k:\phi_{k,V!}\mathcal{G}_V \to \mathcal{F}_{(k)}$ defined by its action on the basis $(\epsilon'_s)_{s \in B_\bullet}$:
	\begin{equation}
	\label{eq:gamma_k}
	\omega_k(\epsilon'_s) = \Phi_{k,s} Y_s,
	\end{equation}
	where $Y_s$ is defined in (\ref{eq:Y_s}) and the unitary factors $\Phi_{k,s}$ are given by
	\begin{equation}
	\label{eq:Phi_ks}
	\Phi_{k,s} := \exp(-2 \pi \ii k \cdot \rho(c_s))
	\end{equation}
	As follows from the identity
	\begin{equation}
	\nonumber
	\Phi_{k,s}\exp(-2 \pi \ii k \cdot t_{s,i}) = \Phi_{k,\delta_{n,i}s} \exp(-2 \pi \ii k \cdot l_{s,i})
	\end{equation}
	$\omega_k$ commutes with the face operators and is therefore an isomorphism of semi-simplicial vector spaces $\phi_{k,V!}\mathcal{G}_V$ and $\mathcal{F}_{(k)}$.
\end{proof}
\begin{proposition}
	\label{prop:dense}
	Let $J$ be a generating set of $d\mbox{-cycles}$ of $\mathcal{G}_V$ (which can always be chosen finite since $\Z[L_V]$ is a Noetherian ring) and let $r_{j,s} \in \Z[L_V]$ be the coefficients of the cycle $j \in J$ in the basis of $\mathcal{G}_{V,d}$:
	$$
	j = \sum_{s \in B_d} r_{j,s} \epsilon_s
	$$
	Then for almost all $k \in V$ with the possible exception of a nowhere dense subset of $V$, the space of $d\mbox{-cycles}$ of $\mathcal{F}_{(k)}$ is spanned by the set of vectors
	\begin{equation}
	\label{eq:gen_set}
	\left\{ \sum_{s\in B_\bullet} \phi_{k,V}(r_{j,s}) \Phi_{k,s} Y_s \right\}_{j \in J},
	\end{equation}
	where the phase factors $\Phi_{k,s}$ are given by (\ref{eq:Phi_ks})). 
\end{proposition}
\begin{corollary}
	\label{cor:smooth}
	For any linear subspace $V \subset E^*$, there exists a finite set of smooth functions $\mathfrak{a}_j: V \to \C^m$ indexed by a generating set of the $d\mbox{-cycles}$ of $\mathcal{G}_V$, such that for almost all 
	$k \in \mathcal{B} \cap V$ (with the possible exception of a subset nowhere dense in $V$), the vector $(a_k(x_1), \dots,a_k(x_m)) \in \C^m$ of partial diffraction amplitudes defined in Proposition \ref{def_a_k} belongs to a subspace spanned by $\mathfrak{a}_j(k)$. 
\end{corollary}
As follows from (\ref{eq:Y_s}), (\ref{eq:Phi_ks}) and (\ref{eq:gen_set}), the components of $\mathfrak{a}_j$ are finite exponential sums of the form
\begin{equation}
\label{eq:trig_sum}
\mathfrak{a}_{j,p}(k) = \sum_{l \in X_{j,p}} C_{lj,p}\exp\left(-2\pi \ii k \cdot(l + y_p)\right),
\end{equation}
where $X_{j,p}$ is a finite subset of $L_V$, $C_{lj,p} \in \Z$ and $y_p \in V^*$. It is remarkable that while the dependence of the diffraction amplitudes on $k\in \mathcal{B} \cap V$ is usually utterly irregular, the constraints (\ref{eq:trig_sum}) for the partial amplitudes are smooth functions of $k$. Note, however, that these constraints are effective only if the Bragg peaks fill densely an open subset of $V$ and when the number $m$ of distinct atomic sites in (\ref{varrho}) is large enough (e.g. if $m$ exceeds the number of generators of the $d\mbox{-cycles}$ of $\mathcal{G}_V$).
\par
To prove Proposition \ref{prop:dense} we shall need the following technical result:
\begin{proposition}
	\label{prop:rank}
	If $\vartheta$ is a homomorphism of free $\Z[L_V]\mbox{-modules}$, then for all $k \in V$
	$$
	\rank\left( \phi_{k,V!}\vartheta \right) \le \rank(\vartheta).
	$$
	Moreover, the subset $\mathcal{V}_\vartheta \subset V$ for which this inequality is strict
	$$
	\mathcal{V}_\vartheta :=\left\{k \in  V : \rank\left( \phi_{k,V!}\vartheta \right) < \rank(\vartheta) \right\}
	$$
	is nowhere dense in $V$.
\end{proposition}
\begin{proof}
	Let $\rank(\vartheta)=r$. Since $\Z[L_V]$ is a commutative domain, the rank of a homomorphism of free modules coincides with that of its matrix over the quotient field of $\Z[L_V]$ \cite[Ex. 5.23A]{lam2009exercises}. Hence, all minors of order $r+1$ of the matrix of $\vartheta$ (if any) are zero, and there exists a non-zero minor of order $r$. All these minors are elements of $\Z[L_V]$ and the corresponding minors of $\phi_{k,V!}\vartheta$ are obtained from them by the ring homomorphism $\phi_{k,V}$. Therefore, $\rank\left( \phi_{k,V!}\vartheta \right) \le r$ for any $k \in V$. Let $0 \neq P \in \Z[L_V]$ be a non-vanishing minor of $\vartheta$ of order $r$. Then the expression
	$$
	k \mapsto \phi_{k,V}(P)
	$$
	defines an entire analytic function on $V \otimes_{\R} \C$ vanishing on $\mathcal{V}_\vartheta$ and thus also on the closure of $\mathcal{V}_\vartheta$ in $V$. If this closure contains an open subset of $V$, this function is zero. Since the set of functions $V \to \C$
	$$
	\left\{ k \mapsto \exp(-2 \pi \ii k\cdot l) \mid l\in L_V\right\}
	$$
	is linearly independent, $P=0$. This contradiction proves the Proposition.
\end{proof}
\begin{proof}[Proof of Proposition \ref{prop:dense}]
	Consider the homomorphism $\varkappa: \Z[L_V]^{(J)} \to \mathcal{G}_{V, d}$ mapping each basis element of the free $\Z[L_V]\mbox{-module}$ over the set $J$ to the corresponding element of $J \subset \mathcal{G}_{V, d}$. Since $J$ is the generating set of the submodule of $d\mbox{-cycles}$, the following sequence of free $\Z[L_V]\mbox{-modules}$
    \begin{equation}
    \label{eq:exact}
    \xymatrix{
	    \Z[L_V]^{(J)} \ar[r]^-\varkappa & \mathcal{G}_{V, d} \ar[r]^-{\partial_d} & \mathcal{G}_{V, d-1}
    }
    \end{equation}
    is exact. Applying the extension of scalars $\phi_{k,V!}$ to (\ref{eq:exact}) yields the upper row of the following diagram (which is commutative by Proposition \ref{prop:iso}):
    \begin{equation}
    \label{eq:diagram}
    \xymatrix@!C{
    	\C^{(J)} \ar@{=}[d] \ar[r]^-{\phi_{k,V!}\varkappa} & \phi_{k,V!}\mathcal{G}_{V, d} \ar[d]^{\omega_k} \ar[r]^-{\phi_{k,V!}\partial_d} & \phi_{k,V!}\mathcal{G}_{V, d-1} \ar[d]^{\omega_k}\\
    	\C^{(J)} \ar[r]^-{\omega_k \circ (\phi_{k,V!}\varkappa)} & \mathcal{F}_{(k), d} \ar[r]^-{\partial_d} & \mathcal{F}_{(k), d-1}
    }
    \end{equation}
    Since the sequence (\ref{eq:exact}) remains exact when extended to the vector spaces over the quotient field, one has 
	$$
	\rank(\varkappa)+\rank(\partial_d) = \#B_d.
	$$ 
	Then by Proposition \ref{prop:rank}, for all $k \in V \backslash \left( \mathcal{V}_\varkappa \cup \mathcal{V}_{\partial_d} \right)$ holds the equality
	$$
	\rank\left(\phi_{k,V!}\varkappa\right)+\rank\left(\phi_{k,V!}\partial_d\right) = \#B_d.
	$$
	and thus the rows of (\ref{eq:diagram}) are also exact. As follows from (\ref{eq:gamma_k}), the vectors of the set (\ref{eq:gen_set}) are images of the the basis vectors of $\C^{(J)}$ in the bottom row of (\ref{eq:diagram}). Therefore, for all $k \in V \backslash \left( \mathcal{V}_\varkappa \cup \mathcal{V}_{\partial_d} \right)$ the subspace of $d\mbox{-cycles}$ in $\mathcal{F}_{(k)}$ is spanned by (\ref{eq:gen_set}). Since $\mathcal{V}_\varkappa \cup \mathcal{V}_{\partial_d}$ is nowhere dense in $V$, this proves the Proposition.
\end{proof}
\par
In the physically relevant case where $V=E^*$ (and therefore $L_V$ can be considered as a subgroup of $E$), is possible to give a geometric meaning to formula (\ref{eq:trig_sum}) in the following way. Since $L_V$ is a quotient of $H_1(B, \Z)$ by the kernel of (\ref{eq:L_V}), there exists a normal semi-simplicial covering $\widetilde{B}_V \to B$ for which $L_V$ is the group of deck transformations. The corresponding action of $L_V$ on $\Z^{(\widetilde{B}_V)}$ commutes with the face homomorphisms and therefore defines on $\Z^{(\widetilde{B}_V)}$ the structure of a semi-simplicial $\Z[L_V]\mbox{-module}$. It is straightforward to check using formula (\ref{eq:face_op}) that this module is isomorphic to $\mathcal{G}_V$. Hence, the $d\mbox{-chains}$ of $\Z^{(\widetilde{B}_V)}$ can be seen as formal finite integer linear combinations of copies of prototiles translated by vectors of $L_V$. The $d\mbox{-cycles}$ of $\Z^{(\widetilde{B}_V)}$ (and thus also those of $\mathcal{G}_V$) then correspond to the combinations with boundaries canceling out. The advantage of this approach is that it can be directly applied to tilings with tiles of arbitrary shapes, without preliminary triangulation, as illustrated by Figure \ref{fig:flip}. This representation can be used directly to calculate the sum at the right hand side of (\ref{eq:trig_sum}). Assuming that Figure \ref{fig:flip} depicts the cycle $j$, and the decoration corresponds to the point $x_p$, the vectors $l+y_p \in E$ in (\ref{eq:trig_sum}) are given by the positions of the copies of the decorating point, and the coefficients $C_{lj,p}$ are the weights of the corresponding tiles ($+1$ and $-1$ for the case shown on Figure \ref{fig:flip}). 
\begin{figure}[h]
	\centering
	\includegraphics[width=0.8\linewidth]{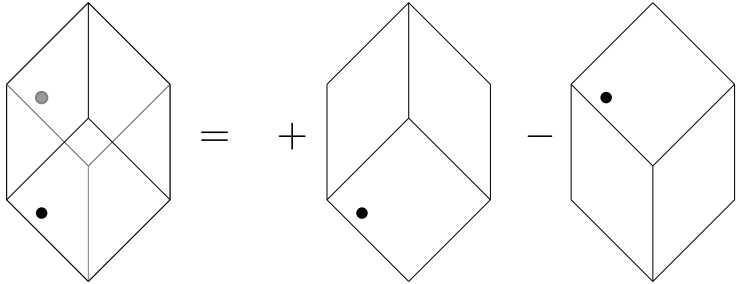}
	\caption{A metaphorical representation of an integer linear combination of translated prototiles corresponding to a $d\mbox{-cycle}$ of $\Z^{(\widetilde{B}_V)}$, for the case of the tiling of plane by squares and 45 degrees rhombi. The solid circles represent the decoration of the square tile.}
	\label{fig:flip}
\end{figure}
\section{Examples}
\label{sec:examples}
\subsection{Binary patterns in one dimension}
One-dimensional binary patterns are decorated tilings of the real line by two types of intervals. Let $v_1$ and $v_2$ stand for the length of the intervals and let $u_1$ and $u_2$ be the positions of the decorating points (relative to the left end of the respective intervals). The diffracting distribution (\ref{eq:density}) of a binary pattern is defined by assigning complex weights $w_1$ and $w_2$ to the respective decorations (see Figure \ref{fig:binary_seq}).  
\begin{figure}[h]
	\centering
	\includegraphics[width=0.6\linewidth]{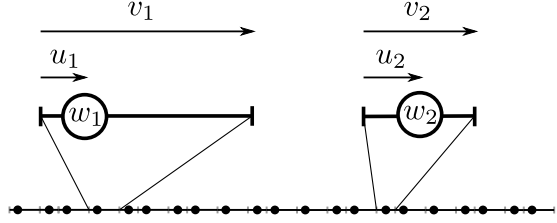}
	\caption{One-dimensional binary pattern.}
	\label{fig:binary_seq}
\end{figure}
Let us now describe the FBS-complex $(B, E, \rho)$ encoding the binary patterns. The graded set $B_\bullet$ contains three elements:
\begin{equation}
\nonumber
\begin{matrix}
B_0 &=& \{s_0\}\\
B_1 &=& \{s_1, s_2\}
\end{matrix}
\end{equation} 
with the face maps given by
\begin{equation}
\nonumber
\delta_{1,0}s_i = \delta_{1,1}s_i = s_0 \qquad \text{ for } i\in\{1,2\}.
\end{equation}
Therefore, the geometric realization $|B|$ of the semi-simplicial set $B$ is a bouquet of two circles (see Figure \ref{fig:fbs_bouquet}). The group of $1\mbox{-cycles}$ of $C_\bullet(B, \Z)$ is generated by the  one-dimensional simplices $s_1$ and $s_2$. Finally, the homomorphism $\rho: H_1(B, \Z) \to E$ is defined by the formula
\begin{equation}
\nonumber
\rho(s_i) = v_i\qquad  \text{ for } i\in\{1,2\}
\end{equation}
(recall that in this case $E$ is a real line).
\par
\begin{figure}[ht]
	\centering
	\includegraphics[width=0.7\linewidth]{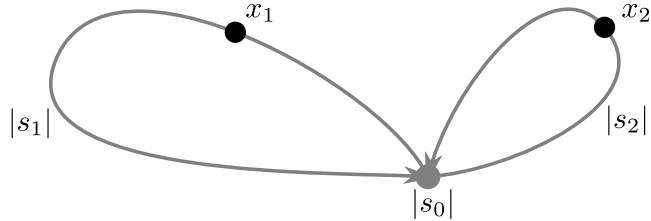}
	\caption{The geometric realization of the FBS-complex of a binary pattern and its decorations.}
	\label{fig:fbs_bouquet}
\end{figure}
Let us now assume that the set of Bragg peaks of the considered binary pattern is dense in $E^*$. With the notation used in of Section \ref{sec:subspace}, this amounts to the assumption that $V=E^*$ with $V^*$ naturally isomorphic to $E$. Therefore, the group $L_V$ (see (\ref{eq:L_V})) is generated by $v_1$ and $v_2$. 
\par
Let us first consider the case when $v_1$ and $v_2$ are commensurate, that is there exists $v \in E$ such that
\begin{equation}
\nonumber
v_i=n_i v
\end{equation}
where $n_1, n_2 \in \Z$ are coprime. In this case $L_V$ is a free abelian group of rank 1 generated by $v$. It is convenient to write the group operation in multiplicative notation, and treat the group ring $\Z[L_V]$ as a ring of Laurent polynomials:
\begin{equation}
\nonumber
\Z[L_V] = \Z[\xi, \xi^{-1}]
\end{equation}
where the multiplication by the indeterminate $\xi$ corresponds to the action of $v \in L_V$. Since $B$ contains only one vertex, one can set in (\ref{eq:c_s}) $c_s=0$ for all $s \in B_\bullet$. Then the face homomorphisms (\ref{eq:face_op}) of $\mathcal{G}_V$ take the following form
\begin{equation}
\nonumber
\begin{matrix}
\delta_{1,0}(\epsilon_{s_i}) &=& \xi^{n_i} \epsilon_{s_0}\\
\delta_{1,1}(\epsilon_{s_i}) &=& \epsilon_{s_0}
\end{matrix}\qquad \text{ for } i\in\{1,2\}
\end{equation}
The generating set $J$ of Proposition \ref{prop:dense} contains only one cycle with the coefficients
\begin{equation}
\nonumber
\begin{matrix}
r_{s_1} &=& \xi^{n_2} - 1 \\
r_{s_2} &=& 1 - \xi^{n_1} 
\end{matrix}
\end{equation}
Therefore, the partial amplitudes $a_k(x_1)$ and  $a_k(x_2)$ for all Bragg peaks with the exception of a nowhere dense set obey the following constraint:
\begin{equation}
\nonumber
\frac{a_k(x_1)}{a_k(x_2)} = 
\exp\left(2 \pi \ii k (u_2-u_1)\right)\frac
{\exp(-2\pi \ii k v n_2) - 1}
{1- \exp(-2\pi \ii k v n_1)}, 
\end{equation}
valid for $k \notin v^{-1}\Z$. In the case $n_1=n_2=1$ and $u_1=u_2$ this amounts to $a_k(x_1)= -a_k(x_2)$, which reflects the trivial fact that for $w_1=w_2$ in this case one recovers the periodic Dirac comb.
\par
When $v_1$ and $v_2$ are incommensurate, $L_V$ is a free abelian group of rank 2 generated by $v_1$ and $v_2$. Again, we shall use the multiplicative notations
\begin{equation}
\nonumber
\Z[L_V] = \Z[\xi_1, \xi_1^{-1}, \xi_2, \xi_2^{-1}]
\end{equation}
with the multiplication by the indeterminate $\xi_i$ corresponding to the action of $v_i$. The face homomorphisms (\ref{eq:face_op}) are then given by
\begin{equation}
\nonumber
\begin{matrix}
\delta_{1,0}(\epsilon_{s_i}) &=& \xi_i \epsilon_{s_0}\\
\delta_{1,1}(\epsilon_{s_i}) &=& \epsilon_{s_0}
\end{matrix}\qquad \text{ for } i\in\{1,2\}
\end{equation}
and again, $J$ contains a single cycle  with the coefficients
\begin{equation}
\nonumber
\begin{matrix}
r_{s_1} &=& \xi_2 - 1 \\
r_{s_2} &=& 1 - \xi_1 
\end{matrix}
\end{equation}
leading to the following constraint
\begin{equation}
\nonumber
\frac{a_k(x_1)}{a_k(x_2)} = 
\exp\left(2 \pi \ii k (u_2-u_1)\right)\frac
{\exp(-2\pi \ii k v_2) - 1}
{1- \exp(-2\pi \ii k v_1)},
\end{equation}
valid for $k \neq 0$.
\subsection{Decorated canonical tilings}
Let us fix $n$ vectors $v_1,\dots,v_n \in E$, such that any $d$ of them are linearly independent over $\R$. The prototiles of an $(n,d)\mbox{-canonical}$ tiling \cite{bodini2010crystallization} of $E$ are parallelotopes with edges $\{v_{i_1},\dots,v_{i_d}\}$ (one prototile for every subset $\{i_1,\dots,i_d\}\subset \{1,\dots,n\}$). For the sake of simplicity we shall limit the consideration to the case when $v_1,\dots,v_n$ are linearly independent over $\Q$ and shall also assume that the Bragg peaks are dense in $E^*$.
\par
The results of this paper are not directly applicable to canonical tilings for $d>1$ since the prototiles are not simplices. However, as a $d\mbox{-dimensional}$ parallelotope can be straightforwardly triangulated  by $d!$ simplices, we shall tacitly assume such triangulation applied to every tile. As follows from (\ref{in_sheaf}), partial diffraction amplitudes at the points belonging to the same simplex are related by a trivial phase factor. Since the triangulation of a prototile is purely formal, the same applies to the points belonging to the same prototile of the canonical tiling. To keep focus on the non-trivial constraints only, we shall therefore consider only the case when each prototile is decorated by a single point at its center (see Figure \ref{fig:canonical}). 
\begin{figure}[h]
	\begin{subfigure}[b]{0.6\textwidth}
		\includegraphics[width=\textwidth]{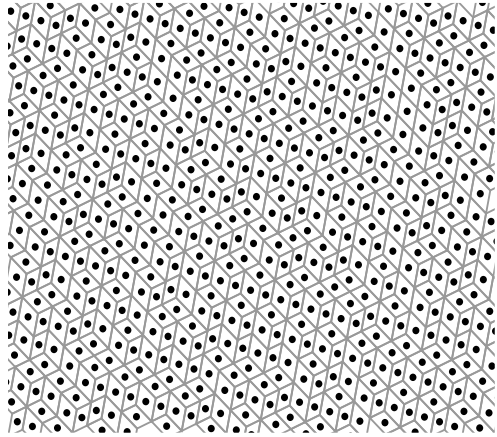}
		\caption{Solid circles represent the decorations, gray lines are the tile boundaries.}
		\label{fig:canonical_1}
	\end{subfigure}\qquad
	\begin{subfigure}[b]{0.28\textwidth}
		\includegraphics[width=\textwidth]{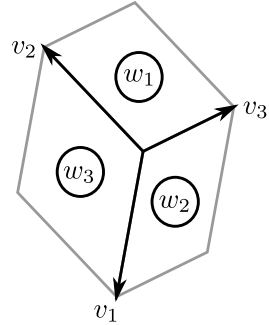}
		\caption{The decorations are positioned at the centers of tiles.}
		\label{fig:canonical_2}
	\end{subfigure}
    \caption{A decorated $(3,2)\mbox{-canonical}$ tiling.}
    \label{fig:canonical}
\end{figure}
\par
Similarly to the previous example, we assume $V=E^*$. The group $L_V \subset E \cong V^*$ is then freely generated by $v_1,\dots,v_n$ and we shall use the multiplicative notation for the ring $\Z[L_V]$:
\begin{equation}
\nonumber
\Z[L_V] \cong \Z[\xi_1, \xi_1^{-1},\dots,\xi_n, \xi_n^{-1}]
\end{equation}
\par
The space $|B|$ of the FBS-complex of an $(n,d)\mbox{-canonical}$ tiling is a (triangulated) $d\mbox{-skeleton}$ of the standard CW-decomposition of an $n\mbox{-dimensional}$ torus $\T^n$. We can use this fact to calculate the homology of the chain complex $(\mathcal{G}_{V, \bullet}, \partial_\bullet)$ in the following way. For $1\le i \le n$, let $(M_{i,\bullet}, \partial_\bullet)$ stand for the chain complex of free $\Z[\xi_i, \xi_i^{-1}]\mbox{-modules}$ of rank 1:
\begin{equation}
\label{eq:circle}
\xymatrix{
	0 \ar[r] & M_{i, 1} \ar[r]^{\xi_i -1} & M_{i,0} \ar[r] & 0
}
\end{equation}
The complex $(M_{i,\bullet}, \partial_\bullet)$ is isomorphic to the complex of submodules of $\mathcal{G}_V$ corresponding to the the edge $v_i$ and the unique vertex of $|B|$. The CW-decomposition of the entire torus then corresponds to the tensor product of $(M_{i,\bullet}, \partial_\bullet)$ over $\Z$ for $i=1\dots n$:
\begin{equation}
\label{eq:complex_M}
(M_\bullet, \partial_\bullet)=\sideset{}{_{\Z}}\bigotimes_{i=1}^n (M_{i,\bullet}, \partial_\bullet)
\end{equation}
Note that (\ref{eq:complex_M}) is naturally a complex of modules over the ring
\begin{equation}
\label{eq:ring}
\sideset{}{_{\Z}}\bigotimes_{i=1}^n \Z[\xi_i, \xi_i^{-1}] \cong \Z[\xi_1, \xi_1^{-1},\dots,\xi_n, \xi_n^{-1}]
\end{equation} 
The $d\mbox{-skeleton}$ of the torus is described by the truncated chain complex
\begin{equation}
\label{eq:truncated}
\xymatrix{
	0 \ar[r] & M_d \ar[r]^{\partial_d} & M_{d-1} \ar[r]^{\partial_{d-1}} & \dots \ar[r]^{\partial_{1}} & M_0 \ar[r] & 0
}
\end{equation}
The triangulation of the prototiles of the canonical tiling thus yields a chain quasi-isomorphism (see e.g. \cite[Chapter~1.1]{weibel1995introduction}) of (\ref{eq:truncated}) to $(\mathcal{G}_V, \partial_\bullet)$. Since the chain complex (\ref{eq:circle}) is acyclic, so is $(M_\bullet, \partial_\bullet)$ and the $d\mbox{-cycles}$ of the truncated complex (\ref{eq:truncated}) are precisely the boundaries of the $(d+1)\mbox{-chains}$ of $(M_\bullet, \partial_\bullet)$. Therefore, the module of $d\mbox{-cycles}$ of $\mathcal{G}_V$ is a free $\Z[L_V]\mbox{-module}$ of rank $\binom{n}{d+1}$. Taking into account that an $(n,d)\mbox{-canonical}$ tiling has $\binom{n}{d}$ prototiles, the constraints of Corollary \ref{cor:smooth} are effective in the case $n\le 2d$.
\par
It is instructive to consider in details the case $n=d+1$. In this situation, the generating set of $d\mbox{-cycles}$ of $\mathcal{G}_V$ contains only one element. Therefore, the partial amplitudes $a_k$ at $d+1$ decorating points must belong to a one-dimensional subspace of $\C^{d+1}$, depending smoothly on $k$. Let us denote by $x_p$ the decorating point belonging to the prototile {\em not} having $v_p$ as its edge (see Figure \ref{fig:canonical_2}). A straightforward computation (using the approach presented at the end of Section \ref{sec:subspace}) then leads to the following formula for the partial amplitudes: 
\begin{equation}
\nonumber
a_k(x_p) = A(k) \sin(\pi k \cdot v_p),
\end{equation}
where the coefficient $A(k)$ depends on the tiling under consideration and is clearly not a regular function of $k$. 
\subsection{Decorated square-triangle tiling}
The square-triangle tiling is a popular model for the structure of planar aperiodic systems with twelve-fold symmetry. It describes remarkably well the quasiperiodic order observed in molecular dynamics simulation \cite{leung1989dodecagonal}, and has also attracted attention recently in soft mater physics \cite{iacovella2011self}. Square-triangle tilings are commonly decorated at vertices, but we shall consider  decorations positioned at the center of tiles instead (see Figure \ref{fig:tiles_st}).
\begin{figure}[h]
	\centering
	\includegraphics[width=0.8\linewidth]{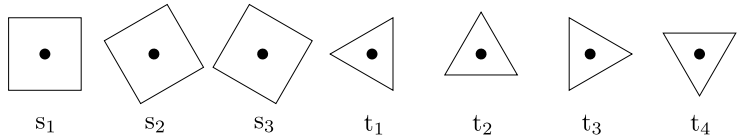}
	\caption{Seven prototiles of the square-triangle tiling. Solid circles represent decorations placed at the centers of prototiles.}
	\label{fig:tiles_st}
\end{figure}
\par
We shall impose no additional constraint on the tiling beyond the assumption that the pure-point part of its diffraction spectrum is dense in $E^*$ (a lot of examples of such tilings can be obtained by some inflation procedure, see e.g. \cite{gahler1997diffraction}). The local order in the tiling is thus described by an FBS-complex obtained by gluing together the (triangulated) prototiles of Figure \ref{fig:tiles_st}. The group $L_V$ is generated by the edges of the tiling and is
a free abelian subgroup  of $E$ of rank 4. The computation performed with the computer algebra system Nemo \cite{Nemo} shows that the null space of the boundary operator of $\mathcal{G}_V$ in degree 2 has rank 2 (over the quotient field of $\Z[L_V]$). Fortunately, it is possible to visualize the generating $2\mbox{-cycles}$ using the approach given at the end of Section \ref{sec:subspace} (see Figure \ref{fig:cycles_st}). 
\begin{figure}[h]
	\centering
	\includegraphics[width=0.6\linewidth]{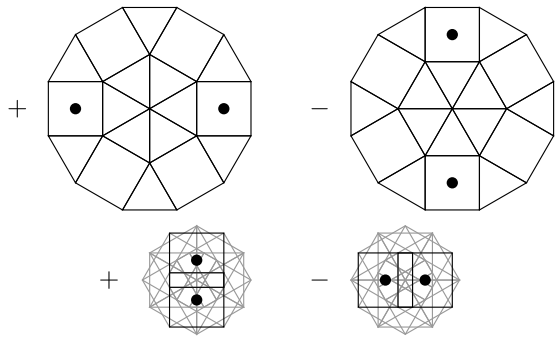}
	\caption{A metaphorical representation of two $2\mbox{-cycles}$ of $\Z^{(\widetilde{B}_V)}$ for the square-triangle tiling (see also Figure \ref{fig:flip}). The solid circle represents the decoration of the tile $\mathrm{s}_1$. Top: the first cycle is a formal difference of two dodecagonal patches of tiling. Bottom: the tiles entering with the same sign in the second cycle are overlapping, but their effective boundary has 12-fold symmetry. The tiles of type $\mathrm{s}_1$ are outlined.}
	\label{fig:cycles_st}
\end{figure}
\par
To give explicit formulas for the constraints on the partial diffraction amplitudes, we shall use the following notation for the twelve vectors corresponding to the edges of the tiling (assumed to be of unit length):
\begin{equation}
\nonumber
l_i=
\begin{pmatrix}
\cos(\pi(i-1)/6)\\
\sin(\pi(i-1)/6)
\end{pmatrix}\qquad\text{ for }i \in \Z/12\Z
\end{equation}
We shall denote the points on $|B|$ corresponding to the decoration of the square $\mathrm{s}_p$ and the triangle $\mathrm{t}_q$ by $x_{\mathrm{s}_p}$ and $x_{\mathrm{t}_q}$ respectively. To obtain the constraints on the partial amplitude at a given decorating point, one has to find all copies of the corresponding tiles on Figure \ref{fig:cycles_st} and take into account their weights and the positions of decorations, as explained at the end of Section \ref{sec:subspace}. This yields the following formulas:
\begin{multline}
\label{eq:squares}
a_k(x_{\mathrm{s}_p}) = 
  A_1(k) \sum_{i=0}^3 (-1)^i \exp\left(\ii\pi (\sqrt{3}+1)k \cdot l_{3i+p}\right)+\\
  A_2(k) \sum_{i=0}^3 (-1)^{i+1} \exp\left(\ii\pi (\sqrt{3}-1)k \cdot l_{3i+p}\right)
\end{multline}
and
\begin{multline}
\label{eq:triangles}
a_k(x_{\mathrm{t}_q}) =\\ 
A_1(k) \sum_{i=0}^2 \left(
    \exp\left(-2\pi\ii\frac{\sqrt{3}}{3}k\cdot l_{4i+q} \right)-
	\exp\left(-2\pi\ii\frac{1+\sqrt{3}}{3}k\cdot l_{4i+q} \right)
\right)+\\
A_2(k) \sum_{i=0}^2 \left(
\exp\left(-2\pi\ii\frac{\sqrt{3}}{3}k\cdot l_{4i+q} \right)-
\exp\left(-2\pi\ii\frac{\sqrt{3}-1}{3}k\cdot l_{4i+q} \right)
\right),
\end{multline}
where the complex-valued functions $A_1(k)$ and $A_2(k)$ depend on the considered tiling. Therefore, for almost all Bragg peaks of a square-triangle tiling, with the possible exception of a set nowhere dense, the seven partial amplitudes (\ref{eq:squares}) and (\ref{eq:triangles}) depend on only two unknown quantities!
\section{Conclusions and discussion}
We have considered the pure-point part of the diffraction spectrum of the families of Delone point patterns in the Euclidean space $E$, obeying local rules in a wide sense of the term (in particular, including disordered systems such as models of decorated random tilings). The partial diffraction amplitudes of such patterns are constrained by linear equations explicitly derivable from the local rules. More specifically, these equations depend on the properties of the corresponding FBS-complex -- a geometric object encoding the local order of the pattern (see Definition \ref{FBS}). Whenever Bragg peaks fill densely a linear subspace $V \subset E$, for almost all of them, with the possible exception of a subset nowhere dense in $V$, the coefficients of these equations depend smoothly on the wave vector $k \in V$. For a given FBS-complex, these coefficients can be calculated explicitly in terms of finite trigonometric sums.   
\par
It has been argued in \cite{kalugin2019robust} that the goal of the structure analysis of aperiodic solids should be the determination of the local environments responsible for the formation of the long range order, rather than finding the position of each and every atom in the structure. The local environments are naturally described by decorated FBS-complexes, and the constraints on partial diffraction amplitudes could be used to evaluate the validity of such structure models. This brings up a question: are partial amplitudes $a_k$ experimentally measurable? Since formally $a_k$ can be derived (up to a common phase factor) through the dependence of the Bragg peak intensities (\ref{eq:peak}) on the weights $w_p$, one can think of using the method of isotopic substitution in neutron diffraction experiments \cite{cornier1993neutron}. However, this approach does not allow to distinguish the contributions of the same chemical element in different local environments. A alternative way to access the partial amplitudes is made possible by the recent progress in the development of phasing algorithms \cite{palatinus2013charge}. Namely, one could separate the contribution of different atomic sites to the diffraction via the segmentation of the reconstructed electronic density (see e.g. Figure 2 of \cite{takakura2007atomic}). Technically, such a segmentation can be performed by means of a watershed algorithm \cite{beucher2018morphological}.  
\par
Corollary \ref{cor:smooth} provides a way to prove for a given set of local rules that the pure point part of the diffraction measure of any decorated tiling respecting these rules cannot be dense everywhere. The case of local rules enforcing periodic tilings is a trivial example of such a situation. An interesting open question is whether there are less trivial tilings for which the absence of everywhere dense pure point diffraction can be proven in this way.
\section*{Appendix A}
Let $\mathbf{\Delta}_\mathrm{inj}$ stand for the small category of finite ordered sets $[n]:=\{0<1<\dots < n\}$ and order preserving injective maps, and $\mathbf{\Delta}_\mathrm{inj}^\mathrm{op}$ be the corresponding opposite category. 
\begin{definition}
A semi-simplicial\footnote{This construction is also called a presimplicial object \cite[Chapter 4.1]{fritsch1990cellular}. We follow here the terminology of \cite[Chapter 8.1]{weibel1995introduction}.} object in a category $\mathcal{C}$ is a functor $\mathscr{S}: \mathbf{\Delta}_\mathrm{inj}^\mathrm{op} \to \mathcal{C}$ (or equivalently a {\em contravariant} functor $\mathbf{\Delta}_\mathrm{inj} \to \mathcal{C}$).
\end{definition}
The functor $\mathscr{S}$ is entirely characterized by its value on objects and morphisms of $\mathbf{\Delta}_\mathrm{inj}^\mathrm{op}$. Thus, the first part of the data defining a semi-simplicial object $\mathscr{S}$ is a sequence $\{\mathscr{S}_n, n \in \N\}$ of objects of $\mathcal{C}$, where we use the notation $\mathscr{S}_n$ for the functorial image $\mathscr{S}[n]$. Let $\delta^{n,i}$ stand for the injective map from $[n-1]$ to $[n]$ missing the element $i \in [n]$. The entire set of morphisms of $\mathbf{\Delta}_\mathrm{inj}$ is a transitive closure of $\delta^{n,i}$ (called elementary {\em coface maps}). Let $\delta_{n,i}:=\mathscr{S}\delta^{n,i}$ stand for the functorial image of $\delta^{n,i}$. Therefore, the second part of the data defining $\mathscr{S}$ is a set of the elementary {\em face morphisms} $\delta_{n,i}: \mathscr{S}_n \to \mathscr{S}_{n-1}$. These morphisms must satisfy the so-called simplicial identity:
\begin{equation}
\label{face_identity}
\delta_{n-1,i} \delta_{n,j} =\delta_{n-1,j-1} \delta_{n, i} \quad \text{ whenever } i<j,
\end{equation}
which follows directly from the identity $\delta^{n,j} \delta^{n-1,i} =\delta^{n,i} \delta^{n-1, j-1}$ in the category $\mathbf{\Delta}_\mathrm{inj}$.
\par
In the case where $\mathcal{C}$ is the category of sets and maps, semi-simplicial objects are called semi-simplicial sets. If $B$ is a semi-simplicial set, the disjoint union of the sets $B_n$ has naturally the structure of an $\N\mbox{-graded}$ set, which we will denote by $B_\bullet$. The semi-simplicial set $B$ is called $d\mbox{-dimensional}$ if $B_d \neq \varnothing$ and $B_n=\varnothing$ for all $n>d$. A finite-dimensional semi-simplicial set is finite if all sets $B_n$ are finite. For a semi-simplicial set $B$ we shall denote by $|B|$ its geometric realization (that is the topological cellular complex with the combinatorial structure given by $B$, see e.g. \cite{milnor1957geometric} and \cite[Chapter~4.2]{fritsch1990cellular}). Similarly, for $s \in B_n$, the notation $|s|$ refers to the corresponding $n\mbox{-dimensional}$ simplicial cell $|s| \subset |B|$. 
\par
Let $e^{n,i}: [1] \to [n]$ stand for the morphism in $\mathbf{\Delta}_\mathrm{inj}$ given by
\begin{eqnarray*}
	e^{n,i}(0)=0\\
	e^{n,i}(1)=i.
\end{eqnarray*}
We define the {\em reference edge} maps $e_{n,i}: B_n \to B_1$ as the functorial image of $e^{n,i}$:
\begin{equation*}
e_{n,i} := Be^{n,i}
\end{equation*}
\par
Another important case of semi-simplicial objects corresponds to the situation when $\mathcal{C}$ is the category of modules over a commutative ring $R$. Such semi-simplicial objects are called semi-simplicial $R\mbox{-modules}$ (or semi-simplicial vector spaces if $R$ is a field). If $\mathcal{M}$ is a semi-simplicial $R\mbox{-module}$, the face operators $\delta_{n,i}$ are $R\mbox{-module}$ homomorphisms:
$$
\delta_{n,i}: \mathcal{M}_n \to \mathcal{M}_{n-1}.
$$
As follows from (\ref{face_identity}), the operators $\partial_n: \mathcal{M}_n \to \mathcal{M}_{n-1}$ defined as
$$
\partial_n = \sum_{i=0}^n (-1)^i \delta_{n,i}
$$
satisfy the equation $\partial_{n-1}\partial_n=0$ and thus make the $\N\mbox{-graded}$ module
$$
\mathcal{M}_\bullet = \bigoplus_{n \in \N} \mathcal{M}_n 
$$ 
into a chain complex $(\mathcal{M}_\bullet, \partial_\bullet)$. 
\par
Given a commutative ring $R$ and a semi-simplicial set $B$, one can construct the free semi-simplicial $R\mbox{-module}$ $R^{(B)}$ by postcomposing $B$ with the free $R\mbox{-module}$ functor $R^{(-)}$. It is noteworthy that the chain complex $({\Z^{(B)}}_\bullet, \partial_\bullet)$ is tautologically isomorphic to the complex of cellular chains of $|B|$ considered as a CW-complex. For this reason we shall use for $({\Z^{(B)}}_\bullet, \partial_\bullet)$  the more traditional notation $C_\bullet(B, \Z)$.
\section*{Appendix B}
Traditionally, the diffraction measure is defined as a Fourier transform of the autocorrelation (or Patterson) measure \cite{baake2004dynamical,baake2013aperiodic} of the diffracting quantity. In this Appendix, we shall show that the distribution $\eta$ defined by the formula (\ref{eta}) equals the Fourier transform of the autocorrelation measure of (\ref{varrho}). 
\par
Let us start by constructing the autocorrelation measure $\gamma$ of the weighted Dirac comb $\varrho_f$ in the sense of the dynamical system $(\X(f_0), E, \mu)$. For a given $f \in \X(f_0)$, let us consider the product measure $\overline{\varrho}_f \times \varrho_f$ on $E^2$. Averaging this measure over the hull $\X(f_0)$ yields a positive translation bounded measure on $E^2$
\begin{equation}
\nonumber
\int_{\X(f_0)} \left(\overline{\varrho}_f \times \varrho_f\right) \dd \mu(f),
\end{equation}
which is invariant with respect to translations of the form $(y_1, y_2) \mapsto (y_1+t, y_2 +t)$. Therefore, there exists a positive translation bounded measure $\gamma$ on $E$ such that for any $\psi \in \mathcal{S}(E^2)$ holds the following
\begin{multline}
\label{gamma}
\int_{E^2}\left(
\int_{\X(f_0)}\overline{\varrho}_f(y_1) \varrho_f(y_2) \dd \mu(f)
\right)\psi(y_1, y_2)\dd y_1 \dd y_2 =\\
\int_E\left(
\int_E \psi(t, y+t) \dd t
\right)\dd \gamma(y)
\end{multline}
It can be shown following the Dworkin's argument \cite{dworkin1993spectral} (see also \cite{baake2016spectral} for a detailed account) that if the dynamical system $(\X(f_0), E, \mu)$ is uniquely ergodic, then the naïve autocorrelation measure of $\varrho_f$ exists and is equal to $\gamma$. The former is defined (see for instance \cite{baake2013aperiodic}) as the Eberlein convolution $\widetilde{\varrho}_f \circledast \varrho_f$, where the symbol $\widetilde{\enskip}$ stands for complex conjugation and changing the sign of the function argument, i.e. $\widetilde{\varrho}_f(y)=\overline{\varrho}_f(-y)$. 
\par
Let us now express the left hand side of (\ref{eta}) in terms of $\gamma$. The right-hand side of (\ref{Gamma}) is a square-integrable function with well-defined values everywhere on $\X(f_0)$ (and not just $\mu\mbox{-almost}$ everywhere). Therefore, taking into account (\ref{varrho}), for any $f \in \X(f_0)$ and any $\varphi_1, \varphi_2 \in \mathcal{S}(E)$ one has
\begin{equation}
\nonumber
\sum_{p,q=1}^m\overline{w_p}w_q\left(\overline{\Gamma_{x_p}(\varphi_1)}\Gamma_{x_q}(\varphi_2)\right)(f) =
\int_{E^2} \overline{\varrho}_f(y_1) \varrho_f(y_2) \overline{\varphi_1}(-y_1) \varphi_2(-y_2)
\dd y_1 \dd y_2
\end{equation}
By integrating this identity over the hull $\X(f_0)$ and taking into account (\ref{gamma}), we get
\begin{equation}
\nonumber
\sum_{p,q = 1}^m
\overline{w_p} w_q \left\langle 
\Gamma_{x_p}(\varphi_1), \Gamma_{x_q}(\varphi_2)
\right\rangle =
\int_E (\widetilde{\varphi}_1 \ast \varphi_2)(-y) \dd \gamma(y)
\end{equation}
Using (\ref{eta}) for the left-hand side and expressing the right-hand side through the Fourier transform of $\gamma$ yields
\begin{equation}
\nonumber
\int_{E^*} \overline{\widehat{\varphi_1}}(k) \widehat{\varphi_2}(k) d\eta(k) 
\equiv \eta(\overline{\widehat{\varphi_1}}\widehat{\varphi_2})=
\widehat{\gamma}(\overline{\widehat{\varphi_1}}\widehat{\varphi_2})
\end{equation}
Therefore, since the functions of the form $\overline{\widehat{\varphi_1}} \widehat{\varphi_2}$ are dense in $\mathcal{S}(E^*)$, we have the identity
\begin{equation}
\nonumber
\eta = \widehat{\gamma}.
\end{equation}
\section*{Acknowledgments}
P.K. thanks Marat Rovinski for fruitful discussions.
\bibliographystyle{unsrt}
\bibliography{draft_rational}

\begin{thebibliography}{10}

\bibitem{kalugin2019robust}
Pavel Kalugin and Andr{\'e} Katz.
\newblock Robust minimal matching rules for quasicrystals.
\newblock {\em Acta Crystallographica Section A: Foundations and Advances},
  75(5):669--693, 2019.

\bibitem{bombieri1986distributions}
E~Bombieri and J~E Taylor.
\newblock Which distributions of matter diffract? {An} initial investigation.
\newblock {\em Le Journal de Physique Colloques}, 47(C3):19--28, 1986.

\bibitem{dworkin1993spectral}
Steven Dworkin.
\newblock Spectral theory and x-ray diffraction.
\newblock {\em Journal of mathematical physics}, 34(7):2965--2967, 1993.

\bibitem{deng2008dworkin}
Xinghua Deng and Robert~V Moody.
\newblock Dworkin’s argument revisited: point processes, dynamics,
  diffraction, and correlations.
\newblock {\em Journal of Geometry and Physics}, 58(4):506--541, 2008.

\bibitem{lenz2017stationary}
Daniel Lenz and Robert~V Moody.
\newblock Stationary processes and pure point diffraction.
\newblock {\em Ergodic Theory and Dynamical Systems}, 37(8):2597, 2017.

\bibitem{baake2016spectral}
Michael Baake and Daniel Lenz.
\newblock Spectral notions of aperiodic order.
\newblock {\em arXiv preprint arXiv:1601.06629}, 2016.

\bibitem{baake2004dynamical}
Michael Baake and Daniel Lenz.
\newblock Dynamical systems on translation bounded measures: Pure point
  dynamical and diffraction spectra.
\newblock {\em Ergodic Theory and Dynamical Systems}, 24(6):1867--1893, 2004.

\bibitem{sadun2008topology}
Lorenzo~A Sadun.
\newblock {\em Topology of Tiling Spaces}, volume~46 of {\em University lecture
  series}.
\newblock American Mathematical Society, Providence, Rhode Island, 2008.

\bibitem{lenz2009continuity}
Daniel Lenz.
\newblock Continuity of eigenfunctions of uniquely ergodic dynamical systems
  and intensity of {Bragg} peaks.
\newblock {\em Communications in mathematical physics}, 287(1):225--258, 2009.

\bibitem{gel2014generalized}
I.M. Gel'fand and N.Y. Vilenkin.
\newblock {\em Generalized Functions: Applications of Harmonic Analysis}.
\newblock Number~4 in Generalized functions. Academic Press, London, 1964.

\bibitem{rosenberg1964}
Milton Rosenberg.
\newblock The square-integrability of matrix-valued functions with respect to a
  non-negative {Hermitian} measure.
\newblock {\em Duke Math. J.}, 31(2):291--298, 06 1964.

\bibitem{rudin1985real}
Walter Rudin.
\newblock {\em Real and complex analysis}.
\newblock McGraw-Hill Book Company, Singapore, 1987.

\bibitem{curry2014sheaves}
Justin~Michael Curry.
\newblock {\em Sheaves, cosheaves and applications}, volume 1249 of {\em
  Publicly Accessible Penn Dissertations}.
\newblock University of Pennsylvania, 2014.

\bibitem{lam2009exercises}
Tsit-Yuen Lam.
\newblock {\em Exercises in Modules and Rings}.
\newblock Problem Books in Mathematics. Springer, New York, 2009.

\bibitem{bodini2010crystallization}
Olivier Bodini, Thomas Fernique, and Damien Regnault.
\newblock Crystallization by stochastic flips.
\newblock In {\em Journal of Physics: Conference Series}, volume 226, page
  012022. IOP Publishing, 2010.

\bibitem{weibel1995introduction}
Charles~A Weibel.
\newblock {\em An introduction to homological algebra}.
\newblock Number~38 in Cambridge Studies in Advanced Mathematics. Cambridge
  University Press, Cambridge, 1994.

\bibitem{leung1989dodecagonal}
Pak~Wo Leung, Christopher~L Henley, and G~V Chester.
\newblock Dodecagonal order in a two-dimensional {Lennard-Jones} system.
\newblock {\em Physical Review B}, 39(1):446--458, 1989.

\bibitem{iacovella2011self}
Christopher~R Iacovella, Aaron~S Keys, and Sharon~C Glotzer.
\newblock Self-assembly of soft-matter quasicrystals and their approximants.
\newblock {\em Proceedings of the National Academy of Sciences},
  108(52):20935--20940, 2011.

\bibitem{gahler1997diffraction}
F~G{\"a}hler and R~Klitzing.
\newblock The diffraction pattern of self-similar tilings.
\newblock {\em NATO ASI Series C Mathematical and Physical Sciences-Advanced
  Study Institute}, 489:141--174, 1997.

\bibitem{Nemo}
Nemo computer algebra package.
\newblock \url{https://nemocas.org/}.

\bibitem{cornier1993neutron}
M~Cornier-Quiquandon, R~Bellissent, Y~Calvayrac, JW~Cahn, D~Gratias, and
  B~Mozer.
\newblock Neutron scattering structural study of {AlCuFe} quasicrystals using
  double isotopic substitution.
\newblock {\em Journal of Non-Crystalline Solids}, 153:10--14, 1993.

\bibitem{palatinus2013charge}
Lukas Palatinus.
\newblock The charge-flipping algorithm in crystallography.
\newblock {\em Acta Crystallographica Section B: Structural Science, Crystal
  Engineering and Materials}, 69(1):1--16, 2013.

\bibitem{takakura2007atomic}
Hiroyuki Takakura, Cesar~Pay Gomez, Akiji Yamamoto, Marc De~Boissieu, and
  An~Pang Tsai.
\newblock Atomic structure of the binary icosahedral {Yb--Cd} quasicrystal.
\newblock {\em Nature materials}, 6(1):58--63, 2007.

\bibitem{beucher2018morphological}
Serge Beucher and Fernand Meyer.
\newblock The morphological approach to segmentation: the watershed
  transformation.
\newblock In {\em Mathematical morphology in image processing}, Optical Science
  and Engineering, pages 433--481. CRC Press, 2018.

\bibitem{fritsch1990cellular}
Rudolf Fritsch and Renzo~A Piccinini.
\newblock {\em Cellular structures in topology}, volume~19 of {\em Cambridge
  Studies in Advanced Mathematics}.
\newblock Cambridge University Press, Cambridge, 1990.

\bibitem{milnor1957geometric}
John Milnor.
\newblock The geometric realization of a semi-simplicial complex.
\newblock {\em Annals of Mathematics}, 65(2):357--362, 1957.

\bibitem{baake2013aperiodic}
Michael Baake and Uwe Grimm.
\newblock {\em Aperiodic Order: Volume 1, A Mathematical Invitation}, volume
  149 of {\em Encyclopedia of Mathematics and its Applications}.
\newblock Cambridge University Press, 2013.

\end{thebibliography}
\end{document}